\documentclass[letterpaper, 10 pt, conference]{ieeeconf}  
\IEEEoverridecommandlockouts                              
                      
\newcommand{\pb}{\pi}              
\newcommand{\dec}{u}               
\newcommand{\decset}{\mathcal{\MakeUppercase{\dec}}}  
\newcommand{\conset}{\mathcal{•}{C}}  
\newcommand{\Pb}{\Pi}              
\newcommand{\cost}{c}              
\newcommand{\nrequality}{N}        
\usepackage{dsfont}
\newcommand{\ones}{\mathds{1}}

\usepackage{amsmath}
\usepackage{amsfonts}
\usepackage{algorithm}
\usepackage{xcolor}
\usepackage{balance} 
\let\labelindent\relax
\usepackage{enumitem}

\usepackage{amsthm}
\usepackage{graphicx}
\usepackage{algorithm}
\usepackage[noend]{algpseudocode}
\usepackage{subcaption}
\usepackage{epstopdf}
\usepackage{balance}
\usepackage{cite}
\usepackage{epstopdf}

\usepackage{tikz} 
\usepackage{pgfplots}
\pgfdeclareplotmark{x}
{%
  \pgfpathmoveto{\pgfqpoint{-.70710678\pgfplotmarksize}{-.70710678\pgfplotmarksize}}%
  \pgfpathlineto{\pgfqpoint{.70710678\pgfplotmarksize}{.70710678\pgfplotmarksize}}%
  \pgfpathmoveto{\pgfqpoint{-.70710678\pgfplotmarksize}{.70710678\pgfplotmarksize}}%
  \pgfpathlineto{\pgfqpoint{.70710678\pgfplotmarksize}{-.70710678\pgfplotmarksize}}%
  \pgfusepathqstroke
}

\newtheorem{theorem}{Theorem}

\newtheorem{problem}{Problem}
\newtheorem{assumption}{Assumption}
\newtheorem{remark}{Remark}

\newtheorem{approach}{Approach}

\title{\LARGE \bf
How to Protect Your Privacy? \\ A Framework for Counter-Adversarial Decision Making
}

\author{In\^{e}s Louren\c{c}o, Robert Mattila, Cristian R.
    Rojas and Bo Wahlberg
    \thanks{This work supported by the Wallenberg AI, Autonomous Systems and Software Program
        (WASP) and the Swedish Research Council under
        contract 2016-06079. The authors are with the Division of Decision and
        Control Systems, School of Electrical Engineering and Computer Science, KTH Royal
        Institute of Technology, Stockholm, Sweden. E-mails: {\tt\footnotesize \{ineslo, rmattila, crro, bo\}@kth.se} }%
}

\begin{document}

\maketitle
\thispagestyle{empty}
\pagestyle{empty}

\definecolor{orange_c}{rgb}{0.8500, 0.3250,  0.0980}
\definecolor{blue_c}{rgb}{0, 0.4470,  0.7410}
\definecolor{green_c}{rgb}{0.4660, 0.6740, 0.1880}
\definecolor{orange_light}{RGB}{235, 152, 78}

\newcommand{\orangedot}{\protect\tikz[baseline=-0.6ex]\protect\node[orange_c, mark size=1.0mm, thick]{\protect\pgfuseplotmark{o}};}
\newcommand{\bluecross}{\protect\tikz[baseline=-0.6ex]\protect\node[blue_c, mark size=1.0mm,thick]{\protect\pgfuseplotmark{x}};}
\newcommand{\greentriag}{\protect\tikz[baseline=-0.6ex]\protect\node[green_c, mark size=1.0mm, thick]{\protect\pgfuseplotmark{triangle}};}
\newcommand{\blackline}{\protect\tikz[baseline=-0.6ex]\protect\node[black, mark size=1.0mm,thick]{\protect\pgfuseplotmark{-}};}
\newcommand{\orangeline}{\protect\tikz[baseline=-0.6ex]\protect\node[orange_light, mark size=1.0mm,very thick]{\protect\pgfuseplotmark{-}};}
\newcommand{\orangesmall}{\protect\tikz[baseline=-0.6ex]\protect\node[orange_light, fill opacity=0.2, draw opacity=0.5, mark size = 0.6mm]{\protect\pgfuseplotmark{*}};}
\newcommand{\blackstar}{\protect\tikz[baseline=-0.6ex]\protect\node[black, fill opacity=1, draw opacity=0.5, mark size = 0.6mm]{\protect\pgfuseplotmark{*}};}
\newcommand{\bluesmall}{\protect\tikz[baseline=-0.6ex]\protect\node[blue, mark size=0.1mm, thick]{\protect\pgfuseplotmark{*}};}

\begin{abstract}

We consider a counter-adversarial sequential decision-making problem where an agent computes its private belief (posterior distribution) of the current state of the world, by filtering private information. According to its private belief, the agent performs an action, which is observed by an adversarial agent. We have recently shown how the adversarial agent can reconstruct the private belief of the decision-making agent via inverse optimization. The main contribution of this paper is a method to obfuscate the private belief of the agent from the adversary, by performing a suboptimal action. The proposed method optimizes the trade-off between obfuscating the private belief and limiting the increase in cost accrued due to taking a suboptimal action. We propose a probabilistic relaxation to obtain a linear optimization problem for solving the trade-off.
In numerical examples, we show that the proposed methods enable the agent to obfuscate its private belief without compromising its cost budget.

\end{abstract}

\section{INTRODUCTION}

A Bayesian agent, henceforth referred to as the \textit{agent} or the \textit{decision maker}, gathers information and uses a filter to compute its posterior distribution over the state of nature. We refer to this posterior distribution as the \emph{private belief} of the agent. Based on its private belief, the agent makes a decision that maximizes its expected utility. The corresponding decision is observed by an adversarial agent, whose objective is to reconstruct the private belief of the decision maker.
A schematic representation of this setup is shown in Figure \ref{fig:scheme1}.

In this paper, we study the counter-adversarial problem of protecting the private belief of the decision maker, by allowing it to make suboptimal decisions. The objective of the decision maker is to prevent the adversarial agent (henceforth, the \emph{adversary}) from accurately estimating its private belief, thus preventing the adversary from predicting its future behavior. At the same time, the decision maker has to limit its increase in cost (due to taking a suboptimal action). A schematic representation of the counter-adversarial setup is shown in Figure \ref{fig:scheme2}.

This counter-adversarial decision-making problem has a vast number of applications, ranging from security of cyber-physical systems to protection of investment strategies. 
One potential area of application is social learning \cite{chamley2004rational, krishnamurthy2016partially}. A number of agents acts sequentially, and each computes a private belief by combining private observations with actions performed by previous agents (the public belief). With our framework, an agent could act in a self-centered way and hide its actual private belief from the other agents. This has potential applications in analysing how social and economic herding occurs. 

Another example is portfolio optimization (e.g., \cite{yin2004markowitz}). An investor decides how to invest its capital in a certain set of stocks, based on public information and its expertise. A competing investor wants to make an informed investment decision as well, but does not have the expertise and knowledge of the main investor. By observing the actions of the main investor, the competitor wants to infer the private belief of this investor. We have recently shown in  \cite{cdc} how this can be done using inverse programming. There are several ways to hide investment decisions, such as using dark pools, delay information or fractional investments. In this paper, we study how the investor can keep its private belief obfuscated from the competitor, by slightly altering its portfolio allocation (while limiting the decrease in its risk-adjusted return). 

%

To summarize, the problem we study in this paper is:
\begin{quote}
    \emph{How should an agent modify its optimal decision in order to not expose its private belief, while limiting its cost increase due to taking a suboptimal decision?}
\end{quote}


The main contributions of this paper are as follows:
\begin{itemize}
	\item 
We propose a class of counter-adversarial decision-making problems. The goal of protecting the private belief of the decision maker is made formal, by defining a range of qualitative and quantitative measures;

	\item 
	The trade-off between obfuscating the private belief and limiting the increase in cost is formulated as an optimization problem;	
	
	\item 
	The aforementioned optimization is, in general, not computationally tractable (especially for large decision systems). We derive a probabilistic relaxation that relies only on linear programming;
	
	\item Lastly, the algorithms are validated and evaluated in numerical experiments. 
We highlight insights and intuition that can be drawn from the experiments.
\end{itemize}

The paper is organized as follows.
Section \ref{sec:preliminaries} provides formal details on how decision-making is modeled and analyses how, and under which circumstances, the agent's privacy can be compromised. Different metrics to measure the success of the agent in protecting its private belief are presented in Section \ref{sec:confuse}. In Section \ref{sec:probabilistic}, a tractable solution to solving the trade-off cost-privacy is derived. Finally, in Section \ref{sec:results}, the performance of the proposed methods is evaluated in numerical experiments.

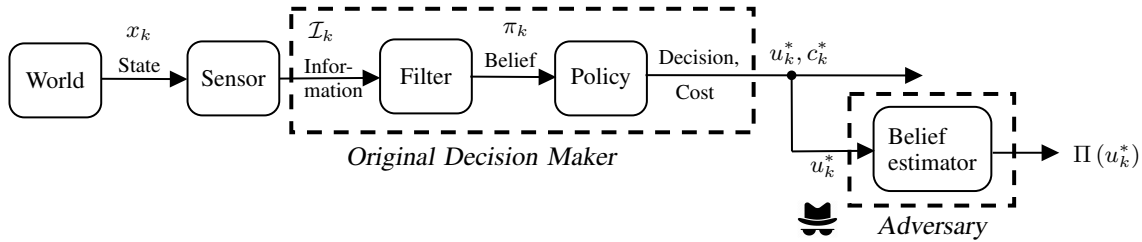
\begin{figure*}[ht!]

\centering

\tikzset{every picture/.style={line width=0.75pt}} 

\begin{tikzpicture}[x=0.75pt,y=0.75pt,yscale=-1,xscale=1]

\draw   (154.33,28.27) .. controls (154.33,24.14) and (157.68,20.8) .. (161.8,20.8) -- (193.13,20.8) .. controls (197.25,20.8) and (200.6,24.14) .. (200.6,28.27) -- (200.6,50.68) .. controls (200.6,54.8) and (197.25,58.15) .. (193.13,58.15) -- (161.8,58.15) .. controls (157.68,58.15) and (154.33,54.8) .. (154.33,50.68) -- cycle ;
\draw  [line width=0.75]  (500.33,62.33) .. controls (500.33,57.64) and (504.14,53.83) .. (508.83,53.83) -- (551.23,53.83) .. controls (555.92,53.83) and (559.73,57.64) .. (559.73,62.33) -- (559.73,87.83) .. controls (559.73,92.53) and (555.92,96.33) .. (551.23,96.33) -- (508.83,96.33) .. controls (504.14,96.33) and (500.33,92.53) .. (500.33,87.83) -- cycle ;
\draw   (64.53,29.27) .. controls (64.53,25.14) and (67.88,21.8) .. (72,21.8) -- (103.33,21.8) .. controls (107.45,21.8) and (110.8,25.14) .. (110.8,29.27) -- (110.8,51.68) .. controls (110.8,55.8) and (107.45,59.15) .. (103.33,59.15) -- (72,59.15) .. controls (67.88,59.15) and (64.53,55.8) .. (64.53,51.68) -- cycle ;
\draw   (339.73,28.27) .. controls (339.73,24.14) and (343.08,20.8) .. (347.2,20.8) -- (378.53,20.8) .. controls (382.65,20.8) and (386,24.14) .. (386,28.27) -- (386,50.68) .. controls (386,54.8) and (382.65,58.15) .. (378.53,58.15) -- (347.2,58.15) .. controls (343.08,58.15) and (339.73,54.8) .. (339.73,50.68) -- cycle ;
\draw    (200.53,39.57) -- (248.57,39.22) ;
\draw [shift={(251.57,39.2)}, rotate = 539.5899999999999] [fill={rgb, 255:red, 0; green, 0; blue, 0 }  ][line width=0.08]  [draw opacity=0] (8.93,-4.29) -- (0,0) -- (8.93,4.29) -- cycle    ;
\draw    (296.93,38.57) -- (337.33,38.81) ;
\draw [shift={(340.33,38.83)}, rotate = 180.35] [fill={rgb, 255:red, 0; green, 0; blue, 0 }  ][line width=0.08]  [draw opacity=0] (8.93,-4.29) -- (0,0) -- (8.93,4.29) -- cycle    ;
\draw  [dash pattern={on 5.63pt off 4.5pt}][line width=1.5]  (488.17,47.33) -- (571.23,47.33) -- (571.23,103.67) -- (488.17,103.67) -- cycle ;
\draw (472.17,110.58) node  {\includegraphics[width=15pt,height=13.88pt]{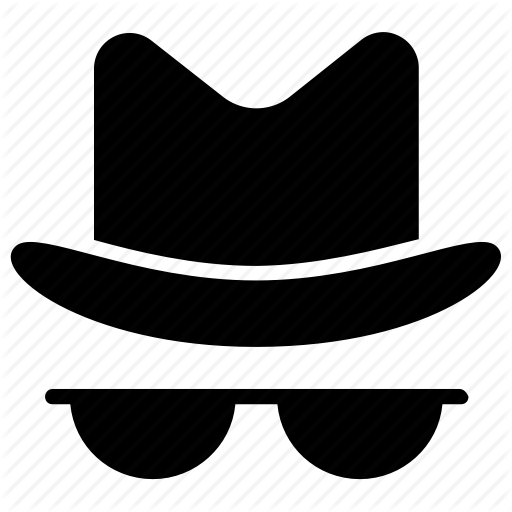}};
\draw [color={rgb, 255:red, 255; green, 255; blue, 255 }  ,draw opacity=1 ][line width=2.25]    (392.33,10.97) -- (499,10.97) ;
\draw [color={rgb, 255:red, 0; green, 0; blue, 0 }  ,draw opacity=1 ]   (385.33,37.83) -- (522.17,38) ;
\draw [shift={(525.17,38)}, rotate = 180.07] [fill={rgb, 255:red, 0; green, 0; blue, 0 }  ,fill opacity=1 ][line width=0.08]  [draw opacity=0] (8.93,-4.29) -- (0,0) -- (8.93,4.29) -- cycle    ;
\draw  [color={rgb, 255:red, 0; green, 0; blue, 0 }  ,draw opacity=1 ][dash pattern={on 5.63pt off 4.5pt}][line width=1.5]  (206.67,4.33) -- (439.83,4.33) -- (439.83,66.67) -- (206.67,66.67) -- cycle ;
\draw    (459.17,76.33) -- (498,76.02) ;
\draw [shift={(501,76)}, rotate = 539.54] [fill={rgb, 255:red, 0; green, 0; blue, 0 }  ][line width=0.08]  [draw opacity=0] (8.93,-4.29) -- (0,0) -- (8.93,4.29) -- cycle    ;
\draw    (560.42,77) -- (591,77) ;
\draw [shift={(594,77)}, rotate = 180] [fill={rgb, 255:red, 0; green, 0; blue, 0 }  ][line width=0.08]  [draw opacity=0] (8.93,-4.29) -- (0,0) -- (8.93,4.29) -- cycle    ;
\draw    (459,38) -- (459.17,76.33) ;
\draw    (111.17,39.67) -- (151.33,39.82) ;
\draw [shift={(154.33,39.83)}, rotate = 180.22] [fill={rgb, 255:red, 0; green, 0; blue, 0 }  ][line width=0.08]  [draw opacity=0] (8.93,-4.29) -- (0,0) -- (8.93,4.29) -- cycle    ;
\draw   (251.44,27.26) .. controls (251.44,23.14) and (254.78,19.79) .. (258.9,19.79) -- (290.23,19.79) .. controls (294.35,19.79) and (297.7,23.14) .. (297.7,27.26) -- (297.7,49.67) .. controls (297.7,53.8) and (294.35,57.14) .. (290.23,57.14) -- (258.9,57.14) .. controls (254.78,57.14) and (251.44,53.8) .. (251.44,49.67) -- cycle ;
\draw  [color={rgb, 255:red, 0; green, 0; blue, 0 }  ,draw opacity=1 ][fill={rgb, 255:red, 0; green, 0; blue, 0 }  ,fill opacity=1 ] (457.08,37.92) .. controls (457.08,36.86) and (457.94,36) .. (459,36) .. controls (460.06,36) and (460.92,36.86) .. (460.92,37.92) .. controls (460.92,38.98) and (460.06,39.83) .. (459,39.83) .. controls (457.94,39.83) and (457.08,38.98) .. (457.08,37.92) -- cycle ;

\draw (87.66,40.47) node  [font=\small,color={rgb, 255:red, 0; green, 0; blue, 0 }  ,opacity=1 ] [align=left] {World};
\draw (177.46,39.47) node  [font=\small,color={rgb, 255:red, 0; green, 0; blue, 0 }  ,opacity=1 ] [align=left] {Sensor};
\draw (530.03,75.08) node  [font=\small,color={rgb, 255:red, 0; green, 0; blue, 0 }  ,opacity=1 ] [align=left] {Belief \\estimator};
\draw (302.86,80.07) node  [color={rgb, 255:red, 0; green, 0; blue, 0 }  ,opacity=1 ,xslant=0.18] [align=left] {Original Decision Maker};
\draw (228.05,39.38) node  [font=\footnotesize] [align=left] {Infor-\\mation};
\draw (317.07,30.2) node  [font=\footnotesize] [align=left] {Belief};
\draw (412.33,29.37) node  [font=\footnotesize] [align=left] {Decision,};
\draw (130.07,32.93) node  [font=\footnotesize] [align=left] {State};
\draw (130,16.8) node  [font=\small]  {$x_{k}$};
\draw (222,15.8) node  [font=\small]  {$\mathcal{I}_{k}$};
\draw (320,14.8) node  [font=\small]  {$\pi _{k}$};
\draw (475,83.8) node  [font=\small]  {$u^{*}_{k}$};
\draw (618,78.8) node  [font=\small]  {$\Pi \left( u^{*}_{k}\right)$};
\draw (530.26,114.87) node  [xslant=0.18] [align=left] {Adversary};
\draw (410.33,46.37) node  [font=\footnotesize] [align=left] {Cost};
\draw (463,26.8) node  [font=\small]  {$u^{*}_{k} ,c^{*}_{k}$};
\draw (274.57,38.47) node  [font=\small,color={rgb, 255:red, 0; green, 0; blue, 0 }  ,opacity=1 ] [align=left] {Filter};
\draw (362.86,39.47) node  [font=\small,color={rgb, 255:red, 0; green, 0; blue, 0 }  ,opacity=1 ] [align=left] {Policy};

\end{tikzpicture}

\caption{Adversarial sequential decision-making scheme. A sensor measures the current world state, $x_k$, and turns it into abstract information $\mathcal{I}_k$. By filtering this information, using \eqref{eq:belief}, an agent denoted \textit{Original Decision Maker} (ODM) obtains a private belief $\pi_k$. Acting rationally according to \eqref{eq:optimization_continuous}, the ODM minimizes a cost function $c(x_k,u_k)$ and performs decision $u_k^*$ with cost $c_k^*$. The adversary observes this decision and reconstructs a set of beliefs, $\Pi(u^*_k)$, using the inverse optimization relation \eqref{eq:set_pi}, that includes the ODM's private belief. The extent to which its privacy is compromised is discussed in Section~\ref{sec:quantifying}.}  
\label{fig:scheme1}
\end{figure*}

\subsection{Related work}

The works \cite{chong2019tutorial, lu2019control, nekouei2019information} present overviews of the topic of privacy for systems and control. 
Methods from different fields are used to address the trade-off between privacy and system performance, such as information theory \cite{farokhi2017fisher}, hypothesis testing \cite{li2016privacy}, and differential privacy \cite{le2013differentially}. The latter started as a method used for static-database applications, but due to its mathematical rigor and strong guarantees made its way to the privacy of dynamical control systems \cite{cortes2016differential}.


In this work, we study privacy in the context of inverse problems. Inverse optimal control is an area that focuses on the problem of reconstructing the optimal cost function for a certain system and policy \cite{kalman1964linear}. 
This framework was recently revisited in  \cite{cdc, mattila2017inverse, krishnamurthy2020adversarial}, where inverse filtering is used to infer different characteristics of the system, such as sensor specifications from the sequence of private beliefs.
Specifically, in \cite{cdc} it was shown how the set of private beliefs consistent with the agent's actions can be estimated. The current paper is an extension of this work to the counter-adversarial setup, where privacy-preserving measures are taken to obfuscate the set of private beliefs.

We propose a probabilistic framework based on a similar concept to that of randomized actions in Markov Decision Processes, covered in \cite{puterman2014markov}. This idea originated from the introduction of ``mixed strategies" in the field of game theory \cite{mckinsey2003introduction}. For approximate methods to convexify the problem, we focus on Monte Carlo integration \cite{davis2007methods}. This method is particularly useful for integration in high-dimensional spaces, since it has been shown to have an accuracy in terms of the standard deviation of the error independent of the number of dimensions.

\begin{figure*}[ht!]

\centering

\tikzset{every picture/.style={line width=0.75pt}} 

\begin{tikzpicture}[x=0.75pt,y=0.75pt,yscale=-1,xscale=1]

\draw [color={rgb, 255:red, 74; green, 144; blue, 226 }  ,draw opacity=1 ]   (601.42,81) -- (632,81) ;
\draw [shift={(635,81)}, rotate = 180] [fill={rgb, 255:red, 74; green, 144; blue, 226 }  ,fill opacity=1 ][line width=0.08]  [draw opacity=0] (8.93,-4.29) -- (0,0) -- (8.93,4.29) -- cycle    ;
\draw  [color={rgb, 255:red, 74; green, 144; blue, 226 }  ,draw opacity=1 ][line width=1.5]  (395.87,33.35) .. controls (395.87,29.25) and (399.19,25.93) .. (403.28,25.93) -- (461.79,25.93) .. controls (465.88,25.93) and (469.2,29.25) .. (469.2,33.35) -- (469.2,55.59) .. controls (469.2,59.68) and (465.88,63) .. (461.79,63) -- (403.28,63) .. controls (399.19,63) and (395.87,59.68) .. (395.87,55.59) -- cycle ;
\draw [color={rgb, 255:red, 74; green, 144; blue, 226 }  ,draw opacity=1 ][line width=1.5]    (469.17,43.33) -- (565.17,43.65) ;
\draw [shift={(569.17,43.67)}, rotate = 180.19] [fill={rgb, 255:red, 74; green, 144; blue, 226 }  ,fill opacity=1 ][line width=0.08]  [draw opacity=0] (11.61,-5.58) -- (0,0) -- (11.61,5.58) -- cycle    ;
\draw  [color={rgb, 255:red, 74; green, 144; blue, 226 }  ,draw opacity=1 ][dash pattern={on 5.63pt off 4.5pt}][line width=1.5]  (368.17,9.33) -- (479.17,9.33) -- (479.17,73) -- (368.17,73) -- cycle ;
\draw [color={rgb, 255:red, 255; green, 255; blue, 255 }  ,draw opacity=1 ][line width=3]    (358.17,1.9) -- (360,72.83) ;
\draw   (119.33,33.27) .. controls (119.33,29.14) and (122.68,25.8) .. (126.8,25.8) -- (158.13,25.8) .. controls (162.25,25.8) and (165.6,29.14) .. (165.6,33.27) -- (165.6,55.68) .. controls (165.6,59.8) and (162.25,63.15) .. (158.13,63.15) -- (126.8,63.15) .. controls (122.68,63.15) and (119.33,59.8) .. (119.33,55.68) -- cycle ;
\draw   (29.53,34.27) .. controls (29.53,30.14) and (32.88,26.8) .. (37,26.8) -- (68.33,26.8) .. controls (72.45,26.8) and (75.8,30.14) .. (75.8,34.27) -- (75.8,56.68) .. controls (75.8,60.8) and (72.45,64.15) .. (68.33,64.15) -- (37,64.15) .. controls (32.88,64.15) and (29.53,60.8) .. (29.53,56.68) -- cycle ;
\draw   (215.44,33.26) .. controls (215.44,29.14) and (218.78,25.79) .. (222.9,25.79) -- (254.23,25.79) .. controls (258.35,25.79) and (261.7,29.14) .. (261.7,33.26) -- (261.7,55.67) .. controls (261.7,59.8) and (258.35,63.14) .. (254.23,63.14) -- (222.9,63.14) .. controls (218.78,63.14) and (215.44,59.8) .. (215.44,55.67) -- cycle ;
\draw   (304.73,33.27) .. controls (304.73,29.14) and (308.08,25.8) .. (312.2,25.8) -- (343.53,25.8) .. controls (347.65,25.8) and (351,29.14) .. (351,33.27) -- (351,55.68) .. controls (351,59.8) and (347.65,63.15) .. (343.53,63.15) -- (312.2,63.15) .. controls (308.08,63.15) and (304.73,59.8) .. (304.73,55.68) -- cycle ;
\draw    (165.53,44.57) -- (213.57,44.22) ;
\draw [shift={(216.57,44.2)}, rotate = 539.5899999999999] [fill={rgb, 255:red, 0; green, 0; blue, 0 }  ][line width=0.08]  [draw opacity=0] (8.93,-4.29) -- (0,0) -- (8.93,4.29) -- cycle    ;
\draw    (261.93,43.57) -- (302.33,43.81) ;
\draw [shift={(305.33,43.83)}, rotate = 180.35] [fill={rgb, 255:red, 0; green, 0; blue, 0 }  ][line width=0.08]  [draw opacity=0] (8.93,-4.29) -- (0,0) -- (8.93,4.29) -- cycle    ;
\draw  [color={rgb, 255:red, 0; green, 0; blue, 0 }  ,draw opacity=1 ][dash pattern={on 5.63pt off 4.5pt}][line width=1.5]  (170.67,9.33) -- (368.25,9.33) -- (368.25,73) -- (170.67,73) -- cycle ;
\draw    (76.17,44.67) -- (116.33,44.82) ;
\draw [shift={(119.33,44.83)}, rotate = 180.22] [fill={rgb, 255:red, 0; green, 0; blue, 0 }  ][line width=0.08]  [draw opacity=0] (8.93,-4.29) -- (0,0) -- (8.93,4.29) -- cycle    ;
\draw [color={rgb, 255:red, 255; green, 255; blue, 255 }  ,draw opacity=1 ][line width=2.25]    (368.33,10.33) -- (368.25,72) ;
\draw    (373.17,43.67) -- (351.6,43.37) ;
\draw [color={rgb, 255:red, 74; green, 144; blue, 226 }  ,draw opacity=1 ][line width=1.5]    (373.17,43.67) -- (392.33,43.39) ;
\draw [shift={(396.33,43.33)}, rotate = 539.1800000000001] [fill={rgb, 255:red, 74; green, 144; blue, 226 }  ,fill opacity=1 ][line width=0.08]  [draw opacity=0] (11.61,-5.58) -- (0,0) -- (11.61,5.58) -- cycle    ;
\draw  [line width=0.75]  (541.33,66.6) .. controls (541.33,61.76) and (545.26,57.83) .. (550.1,57.83) -- (591.96,57.83) .. controls (596.8,57.83) and (600.73,61.76) .. (600.73,66.6) -- (600.73,92.9) .. controls (600.73,97.74) and (596.8,101.67) .. (591.96,101.67) -- (550.1,101.67) .. controls (545.26,101.67) and (541.33,97.74) .. (541.33,92.9) -- cycle ;
\draw  [dash pattern={on 5.63pt off 4.5pt}][line width=1.5]  (529.17,51.33) -- (612.23,51.33) -- (612.23,109.67) -- (529.17,109.67) -- cycle ;
\draw (515.17,113.92) node  {\includegraphics[width=15pt,height=13.88pt]{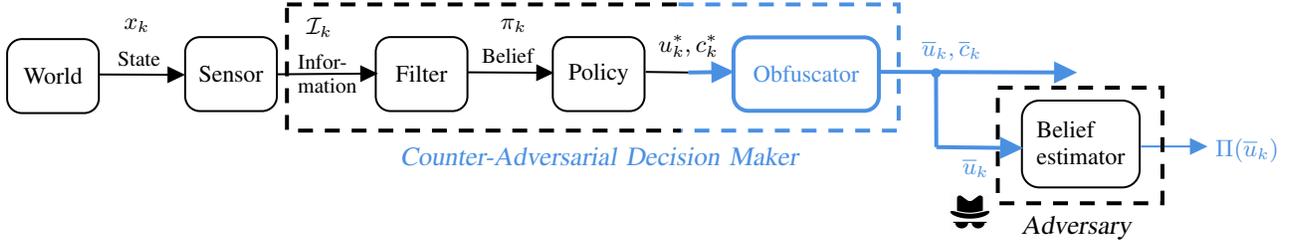}};
\draw [color={rgb, 255:red, 74; green, 144; blue, 226 }  ,draw opacity=1 ][line width=1.5]    (498.17,81.33) -- (538,81.79) ;
\draw [shift={(542,81.83)}, rotate = 180.65] [fill={rgb, 255:red, 74; green, 144; blue, 226 }  ,fill opacity=1 ][line width=0.08]  [draw opacity=0] (11.61,-5.58) -- (0,0) -- (11.61,5.58) -- cycle    ;
\draw [color={rgb, 255:red, 74; green, 144; blue, 226 }  ,draw opacity=1 ][line width=1.5]    (498,43) -- (498.17,81.33) ;
\draw  [color={rgb, 255:red, 74; green, 144; blue, 226 }  ,draw opacity=1 ][fill={rgb, 255:red, 74; green, 144; blue, 226 }  ,fill opacity=1 ] (496.08,43.92) .. controls (496.08,42.86) and (496.94,42) .. (498,42) .. controls (499.06,42) and (499.92,42.86) .. (499.92,43.92) .. controls (499.92,44.98) and (499.06,45.83) .. (498,45.83) .. controls (496.94,45.83) and (496.08,44.98) .. (496.08,43.92) -- cycle ;

\draw (52.66,45.47) node  [font=\small,color={rgb, 255:red, 0; green, 0; blue, 0 }  ,opacity=1 ] [align=left] {World};
\draw (142.46,44.47) node  [font=\small,color={rgb, 255:red, 0; green, 0; blue, 0 }  ,opacity=1 ] [align=left] {Sensor};
\draw (238.57,44.47) node  [font=\small,color={rgb, 255:red, 0; green, 0; blue, 0 }  ,opacity=1 ] [align=left] {Filter};
\draw (327.86,44.47) node  [font=\small,color={rgb, 255:red, 0; green, 0; blue, 0 }  ,opacity=1 ] [align=left] {Policy};
\draw (328.86,86.07) node  [color={rgb, 255:red, 74; green, 144; blue, 226 }  ,opacity=1 ,xslant=0.18] [align=left] {\textcolor[rgb]{0.29,0.56,0.89}{Counter-Adversarial Decision Maker}};
\draw (191.05,44.38) node  [font=\footnotesize] [align=left] {Infor-\\mation};
\draw (282.07,35.2) node  [font=\footnotesize] [align=left] {Belief};
\draw (96.07,36.93) node  [font=\footnotesize] [align=left] {State};
\draw (95,19.8) node  [font=\small]  {$x_{k}$};
\draw (187,20.8) node  [font=\small]  {$\mathcal{I}_{k}$};
\draw (285,19.8) node  [font=\small]  {$\pi _{k}$};
\draw (373.17,30.33) node  [font=\small]  {$u^{*}_{k} ,c^{*}_{k}$};
\draw (571.03,79.08) node  [font=\small,color={rgb, 255:red, 0; green, 0; blue, 0 }  ,opacity=1 ] [align=left] {Belief \\estimator};
\draw (655,82.8) node  [font=\small,color={rgb, 255:red, 74; green, 144; blue, 226 }  ,opacity=1 ]  {$\Pi (\overline{u}_{k})$};
\draw (569.26,121.87) node  [xslant=0.18] [align=left] {Adversary};
\draw (432.53,44.47) node  [font=\small,color={rgb, 255:red, 74; green, 144; blue, 226 }  ,opacity=1 ] [align=left] {Obfuscator};
\draw (518,91.8) node  [font=\small,color={rgb, 255:red, 74; green, 144; blue, 226 }  ,opacity=1 ]  {$\mathnormal{\overline{u}_{k}}$};
\draw (506,32.8) node  [font=\small,color={rgb, 255:red, 74; green, 144; blue, 226 }  ,opacity=1 ]  {$\mathnormal{\overline{u}_{k}} ,\mathnormal{\overline{c}_{k}}$};

\end{tikzpicture}

\caption{A \textit{Counter-adversarial Decision Maker} (CDM) uses an \textit{obfuscator} block to transform its decision $u_k^*$ into a suboptimal decision $\bar{u}_k$ with cost $\bar{c}_k$, by optimizing a privacy measure $\Psi(\bar{u}_k)$ defined in \eqref{eq:ap1}. 
The adversary observes the decision (now $\bar{u}_k$) and, again, reconstructs a set of beliefs (now $\Pi(\bar{u}_k)$). The new decision $\bar{u}_k$ is chosen such that the CDM's privacy is not compromised (according to Section \ref{sec:quantifying}). } 
\label{fig:scheme2}
\end{figure*}
\section{Sequential decision making framework}
\label{sec:preliminaries}
 
In this section, we define our notation and analyse the framework presented in Figure \ref{fig:scheme1} in two separate components. First, we introduce the model under which sequential decisions are made by the decision maker. Then, we provide the relevant details from previous work on how the adversary estimates the decision maker's private belief. We conclude the section with an application example of this framework.
 
\subsection{Notation}

All vectors are column vectors and inequalities between vectors are considered element-wise. The $i$th element of a vector $v$ is $[v]_i$. A probability density function is denoted as $p(\cdot)$, the vector of ones as $\ones$, and the set of positive real numbers as $\mathbb{R}^+$. The distance between the sets $C$ and $D$ is defined as $\text{dist}(C, D) = \inf\{ \left\lVert x-y \right\rVert | x \in C, y \in D \}$. Throughout this paper, we use the terms \textit{decision} and \textit{action} interchangeably.

\subsection{How the decision maker acts}
\label{sec:us}

In this section, we analyse the \textit{Original Decision Maker} (ODM) in Figure \ref{fig:scheme1}. 
The world, or environment, is described by a sequence of states modeled as a random variable $x_k \in \mathcal{X}$, where $\mathcal{X}$ is the state-space and $k$ represents discrete time. At each time instant $k$, some information $\mathcal{I}_k$ is collected from the environment. In order to make a decision, the agent filters the information $\mathcal{I}_k$, to obtain a probabilistic rating over the different states of nature -- its private belief:
\begin{equation}
 \pb_k(x) = p(\, x_k = x \; | \; \mathcal{I}_k \,),
 \label{eq:belief}
 \end{equation}
where $p(\cdot)$ is a conditional density function.

The agent can act on the world by making a decision $u_k$.
Each decision has an associated cost $c(x_k,u_k)$, that also depends on the state of the world. Assuming that it acts rationally, the agent selects the action that minimizes its expected cost under its current belief $\pi_k(x)$ of the world:
\begin{align}
\label{eq:optimization_continuous}
   \min_{\dec_k \in \decset} \quad    & \mathbb{E}_{x_k} \big\{ \, \cost(x_k, \dec_k) \; | \; \mathcal{I}_k \,  \big\} \\
      \text{s.t.}\quad& {\dec_k} \in \conset, \notag
\end{align}
where $\conset$ is the set of feasible actions. 
To yield a tractable analysis, $\mathcal{X}$ is assumed to be discrete ($\mathcal{X} = \{1,\dots,X\}$) and the agent's actions, $u_k \in \mathcal{U}$ where $\mathcal{U} = \mathbb{R}^U$ is the decision set, are assumed to obey affine constrains. Then, \eqref{eq:optimization_continuous} is equivalent to:
%
\begin{align}
\label{eq:forward}
    \min_{\dec_k \in \decset} \quad & \sum_{i=1}^X [\pb_k]_i \cost(i, \dec_k) \\
    \text{s.t.}\quad& A u_k \leq b,  \notag
\end{align}
where $A$ and $b$ define the affine constraints. The resulting resulting optimal decision for the ODM is $u_k^*$, with an associated cost $c^*_k = \mathbb{E}_{x_k} \big\{ \, \cost(x_k, \dec_k^*) \; | \; \mathcal{I}_k \,  \big\} = \sum_{i=1}^X [\pb_k]_i \cost(i, \dec_k^*)$. The decision $u_k^*$ made by the ODM is then publicly available.
%


\subsection{What the adversary can infer from the agent's actions}
\label{sec:adversary}

This section analyses the second component of Figure \ref{fig:scheme1}, the \textit{Adversary} block. 
We say that the privacy of the decision maker is compromised, if, by observing its decision $\dec_k^*$, the adversary is able to estimate the agent's private belief $\pi_k$. It has recently been shown how the private belief can be exposed from the actions. 
In \cite{cdc}, it is shown that if the agent's constraint set $\conset$ is affine and its cost function $c(x_k,u_k)$ is convex and known to the adversary, the latter is able to reconstruct a set of private beliefs, $\Pi(u_k^*)$, consistent with the decision $u_k^*$ made by the agent. This is summarized in the following theorem.
\begin{theorem}[Set of consistent private beliefs $\Pi$, from \cite{cdc}]
The ODM could only have made decision $u_k^*$, if and only if its private belief $\pi_k$ lies in the set $\Pi(u_k^*)$, defined by:
\label{th:set_pi}
 \begin{align}
 \label{eq:set_pi}
        \Pb(u_k^*) =& \\
               &\hspace{-1cm}  \notag \left\{ \pb \in \mathbb{R}^X :  
                 \begin{aligned}
           \exists& \, \lambda  \in \mathcal{R}^{\MakeUppercase{\dec}}, \; \nu \in  \mathcal{R}^{\nrequality} \text{ s.t. } \\
 &  \quad  \pb^T \ones = 1 , \; \pb \geq 0, \;
            \lambda \geq 0, \\
         &   \quad  [\lambda]_j = 0 \text{ if } [\dec_k^*]_j \neq 0, \quad j = 1, \dots, \MakeUppercase{\dec},
            \\[-0.1cm]
       &    \quad  \sum_{i=1}^X [\pb]_i \nabla_\dec \cost(i, \dec_k^*) - \lambda
            + A^T \nu = 0. 
    \end{aligned}   \right\}.
    \end{align}
\end{theorem}
\begin{proof}[Proof (Outline)]
The theorem follows by deriving the Karush-Kahn-Tucker conditions for \eqref{eq:forward} and considering the ODM's private belief $\pi_k$ as an unknown variable. Full details are available in \cite{cdc}.
\end{proof}

\begin{remark} [Estimating cost functions]
Although Theorem \ref{th:set_pi} was derived for the case where the cost function is assumed to be known by the adversary, it can also be applied in cases where it is unknown but can be estimated. 
For example, using \textnormal{revealed preferences} the cost function represents the agent's preferences and is estimated from its choices \cite{afriat1967construction, varian2012revealed}. 
\end{remark}

\subsection{Example -- Portfolio optimization}
\label{sec:portfolio_allocation}

We now exemplify of how the two components described can be applied in an investment scenario, summarizing the complete framework from Figure \ref{fig:scheme1}.

A Markowitz-type investor has access to some private information, $\mathcal{I}_k$. Based on this information, it estimates the optimal investment allocation, with respect to maximizing the risk-adjusted return in a regime-switching market scenario \cite{yin2004markowitz}, \cite{boyd2017multi}. 
The optimization problem \eqref{eq:optimization_continuous} the ODM (i.e., the investor) solves, at each timestep, is the following:
  \begin{align}
    \min_{u_k \in \mathcal{R}^U} & \quad \mathbb{E}_{x_k}\left\{ \gamma \dec_k^T
    \Sigma_{x_k} \dec_k - {\mu}_{x_k}^T \dec_k \;|\; \mathcal{I}_k \right\} \label{eq:mpt_optimization} \\
    \text{s.t.} \; & \quad \ones^T \dec_k = 1, \; \dec_k \geq 0, \notag
\end{align}
where $u_k \in \mathbb{R}^U$ is the portfolio allocation vector and $[u_{k}]_i$ represents the fraction of the capital that is invested in stock $i$ at timestep $k$. The mean vector $\mu_{x_k} \in \mathbb{R}^U$ and covariance matrix
$\Sigma_{x_k} \in \mathbb{R}^{U \times U}$ can be computed according to the state of the market conditions, $x_k$. The risk aversion parameter $\gamma$ quantifies the trade-off between the two terms of the expected value, which are how much return is expected \textit{versus} how risky the investment is.
The constraint $\ones^T \dec_k = 1$ implies that all the money has to be allocated, and $\dec_k \geq 0$ enforces that the investor is allowed to buy but not sell stocks.

Crucial to the success of any investor is how well it can estimate the current market conditions $x_k$. This estimate depends on the private information set $\mathcal{I}_k$ available to each investor. In our scenario, we assume that a rival investor (i.e., an adversary) has access to less (or worse) private information and aims to estimate the private belief of the investor solving \eqref{eq:mpt_optimization}.
In Section \ref{sec:adversary} it was shown how the rival can do this, and in the next section we propose a framework for how the investor can prevent it.

\section{Counter-adversarial decision-making}
\label{sec:confuse}


In this section, we propose a privacy-preserving decision maker that modifies its actions in order to obfuscate its private belief from the adversary. The new framework is shown in Figure \ref{fig:scheme2}.

\subsection{General setup}
\label{sec:general}

At each timestep, the ODM collects information from the world and performs an action (Section \ref{sec:us}). It asks itself: \textit{If I publicly announce the decision $\dec_k^*$, what is the set of beliefs $\Pi(u_k^*)$ consistent with my decision that the adversary can determine?} It was seen in Section \ref{sec:adversary} that this set includes the actual private belief, $\pi_k$, and, therefore, the privacy of the ODM can be compromised.

In this section, we propose a \textit{Counter-adversarial Decision Maker} (CDM), that uses the blue block in Figure \ref{fig:scheme2}, called an \textit{Obfuscator}, to conceal its private belief from the adversary. While the ODM performs the optimal action $u_k^*$ with cost $c_k^*$, the CDM performs a suboptimal action $\bar{u}_k$ with cost $\bar{c}_k$. 


\subsection{Problem formulation}
\label{sec:prob_form}

Performing a suboptimal action $\bar{u}_k $ entails an increase in cost ($\bar{c}_k \geq c_k^*$). Thus, protecting its privacy comprises a trade-off between how much the decision maker is able to obscure its private belief \textit{versus} how much it is willing to pay for doing so, due to performing a suboptimal action. This can be formulated as: 

\begin{problem}[Obfuscating the private belief]
\label{pr:main}
How can an agent obfuscate its private belief $\pb_k$ from an adversary, while keeping its cost as small as possible?
\end{problem}

A straightforward approach for solving Problem~\ref{pr:main}, that explicitly reflects the trade-off, is the following: 
\begin{approach}[Direct way of obfuscating the private belief]
\label{ap:1}
The CDM addresses Problem \ref{pr:main} by solving the following optimization problem:
\begin{align}
    \max_{\bar{u}_k \in \decset} \quad & \Psi(\bar{u}_k)  \label{eq:ap1}  \\
    \text{s.t.}\quad & \bar{u}_k \in \conset, \notag \\
    & \mathbb{E}_{x_k} \big\{ \, \cost(x_k, \bar{u}_k) \big\} \leq c^*_k (1 + c_b).     \notag
\end{align}
\end{approach}
\noindent
The objective function $\Psi(\bar{u}_k)$ quantifies the decision maker's \textit{privacy level}. The last constraint corresponds to how much it is willing to pay, $\bar{c}_k$, to protect its private belief $\pi_k$, and the term $c_b \in \mathbb{R}^+$ is its \textit{obfuscation cost budget} -- the cost the agent allocates to obfuscating its private belief.

\subsection{Quantification of the level of privacy}
\label{sec:quantifying}

Measures to quantify attacks have, for a long time, been a central aspect in the fields of privacy and security. 
In this paper, we quantify it by means of evaluating the privacy level of the defender (i.e., the decision maker), according to one of the measures presented next.
Obfuscating the decision maker's private belief consists of making a suboptimal decision $\bar{u}_k$, such that the set of private beliefs reconstructed by the adversary, $\Pi(\bar{u}_k)$, from \eqref{eq:set_pi}, maximizes one of the following privacy measures:
\begin{enumerate}[label=\alph*)]

\item \textit{Infeasibility}
\begin{align}
    \Psi(\bar{u}_k) = 
    \begin{cases}
0, \quad & \text{if }  \Pi(\bar{u}_k) = \emptyset, \\
-\infty, & \text{otherwise}.
\end{cases} 
\end{align}
By optimizing this objective function, the agent takes a decision that makes the adversary's reconstruction infeasible. Although its privacy is preserved, the adversary might realise that the agent is obfuscating its belief.

\item \textit{Non-uniqueness feasibility}
\begin{align}
    \Psi(\bar{u}_k) = 
    \begin{cases}
0, \quad & \text{if } \exists z \in \Pi(\bar{u}_k) \text{ s.t. } z \neq \pi_k, \\
-\infty, & \text{otherwise}.
\end{cases} 
\end{align}
Under this criterion, the reconstructed set cannot be a unique point. However, even though the adversary is supposedly not able to \textit{uniquely} determine the agent's private belief, additional knowledge of the agent's model might allow it to more confidently identify the actual private belief (see \cite{mattila2019smoother}).

\item  \label{item:c} \textit{Non-existence feasibility}
\begin{align}
    \Psi(\bar{u}_k) = 
    \begin{cases}
0, \quad & \text{if }  \{\pi_k\} \cap \Pi(\bar{u}_k) = \emptyset, \\
-\infty, & \text{otherwise}.
\end{cases} 
\end{align}
This criterion states that the reconstructed set cannot include the actual private belief, $\pi_k \notin \Pi(\bar{u}_k)$. It can either be a single element different from the private belief, or a set of elements that does not include it.

\item \label{item:e} \textit{Desired obfuscation}
\begin{align}
    \Psi(\bar{u}_k) = - \text{dist}(\pi_k^d,\Pi(\bar{u}_k)).
\end{align}
In this case, the agent wants the adversary's set of beliefs to be as close as possible to a desired belief $\pi_k^d$. For example, in the portfolio optimization case, if the agent's private belief is that the market conditions are improving, it might want the adversary to believe they are declining.


\item \label{item:d} \textit{Maximal obfuscation}
\begin{align}
    \Psi(\bar{u}_k) =  \text{dist}(\pi_k,\Pi(\bar{u}_k)).
\end{align}
This criterion states that the agent wants the reconstructed set to be as distant as possible from the actual private belief, which is a generalization of Criterion \ref{item:c}.

\end{enumerate}

\subsection{Computational tractability}  
\label{sec:main_result}


In general, problem \eqref{eq:ap1} is tractable if $\Psi(\bar{u}_k)$ is a concave function. From the different measures of privacy that we presented in the previous section, it can be seen that requiring concavity is typically a strong and unrealistic condition to impose.
Motivated by this, in the next section we present a probabilistic relaxation of \eqref{eq:ap1} that relies only on sampling and linear programming.





\section{Probabilistic framework}
\label{sec:probabilistic}

The privacy-preserving decision $\bar{u}_k$ performed by the CDM requires it to solve the generally intractable optimization problem \eqref{eq:ap1}. In this section, we propose a sampling-based probabilistic relaxation to yield a computationally feasible algorithm to obtain $\bar{u}_k$.

\subsection{Probabilistic relaxation of \textnormal{Approach \ref{ap:1}}}

%
In the context of game theory, we say that the CDM in Section \ref{sec:prob_form} follows a \textit{pure strategy}, in the sense that it solves a deterministic decision problem. Since this strategy could be uncovered by the adversary, a way to obfuscate is to use a \textit{mixed strategy} instead -- to introduce a chance element in the agent's decision process. This idea is similar to using randomized actions in Markov decision processes \cite{puterman2014markov}.

Motivated by this, in this section we propose a \textit{Probabilistic counter-adversarial Decision Maker} (PDM), that uses a mixed strategy for the \textit{obfuscator} block. The randomization involves assigning a probability to each possible action, $\bar{u}_k$, which is now a sample from a distribution. 
The objective of the PDM is to maximize its \textit{expected} privacy level, while limiting its \textit{expected} cost. Unlike the CDM, that performs actions that obfuscate its private belief from the adversary at every timestep, the PDM's actions obfuscate its private belief \textit{on average}.

\begin{approach}[Obfuscation of the private belief on average]
\label{ap:2}
The PDM's approach to Problem \ref{pr:main} consists of sampling its actions from the distribution $\bar{u}_k \sim p_{\bar{u}_k}(\cdot)$ over $\mathbb{R}^U$:
\begin{align}
            \max_{p_{\bar{u}_k}}  \quad    & \mathbb{E}_{\bar{u}_k \sim p_{{\bar{u}_k}} } \big\{ \Psi(\bar{u}_k) \big\} \label{eq:ap2} \\
            \text{s.t.}\quad & \textnormal{support}(p_{{\bar{u}_k}}) \subseteq \conset, \notag \\
            & \mathbb{E}_{\bar{u}_k \sim p_{{\bar{u}_k}} } \big\{ \mathbb{E}_{x_k \sim \pi_k } \big\{ \cost(x_k, \bar{u}_k) \big\}  \big\}  \leq c^*_k( 1 + c_b).  \notag
\end{align}
\end{approach}

\subsection{Computationally feasible formulation}

The optimization problem \eqref{eq:ap2} is still intractable since the PDM 
optimizes over the infinite space of probability distributions in $\mathbb{R}^U$. To make it tractable, we use the Monte Carlo integration technique \cite{davis2007methods}. This technique consists of performing numerical integration using random samples, which in our case gives rise to the following assumption: 

\begin{assumption}
The distribution $p_{\bar{u}_k}(\cdot)$ is assumed to be of the form:
\begin{align}
p_{\bar{u}_k}(\bar{u}_k) = \sum_{l=1}^M  [p]_l \delta (\bar{u}_k-\bar{u}_k^{(l)} ),           
\label{eq:assumption1}
\end{align}
where $\delta$ is the Dirac delta, M is the number of samples, $p \in \mathbb{R}^M$ is a probability mass vector (i.e., $[p]_l \geq 0$ and $\sum_{l=1}^M [p]_l = 1$), and $\{\bar{u}_k^{(l)}\}_{l=1}^M$ is a set in $\mathcal{C}$.
\label{ass:distr}
\end{assumption}

Under Assumption \ref{ass:distr}, the optimization problem \eqref{eq:ap2} reduces to the relaxed problem solved by the PDM:
  \begin{align}
            \max_{ p \in \mathbb{R}^M }  \quad  & \sum_{l=1}^M \; [p]_l \Psi(\bar{u}_k^{(l)})  \label{eq:assump} \\
            \text{s.t.}\quad& [p]_l = 0 \text{ if } \bar{u}_k^{(l)} \notin \conset,  \notag  \\
            & \sum_{l=1}^M \; [p]_l \; \big\{ \sum_{i=1}^X \pi_i \cost(i, \bar{u}_k^{(l)}) \big\} \leq c^*(1 + c_b), \notag  \\
            & [p]_l \geq 0, \quad l = 1, \dots, M, \notag  \\
            & \sum_{l=1}^M [p]_l = 1,  \notag
  \end{align}
which is a finite-dimensional linear program and, therefore, computationally efficient to solve using existing solvers. Under appropriate assumptions on how the points $\{ \bar{u}_k^{(l)} \}_{l=1}^M$ are selected, the solution of \eqref{eq:assump} will tend to that of \eqref{eq:ap2} as $M$ tends to infinity. 
%
%


 \begin{figure}[t]
	\centering
	\includegraphics[width=0.4\textwidth]{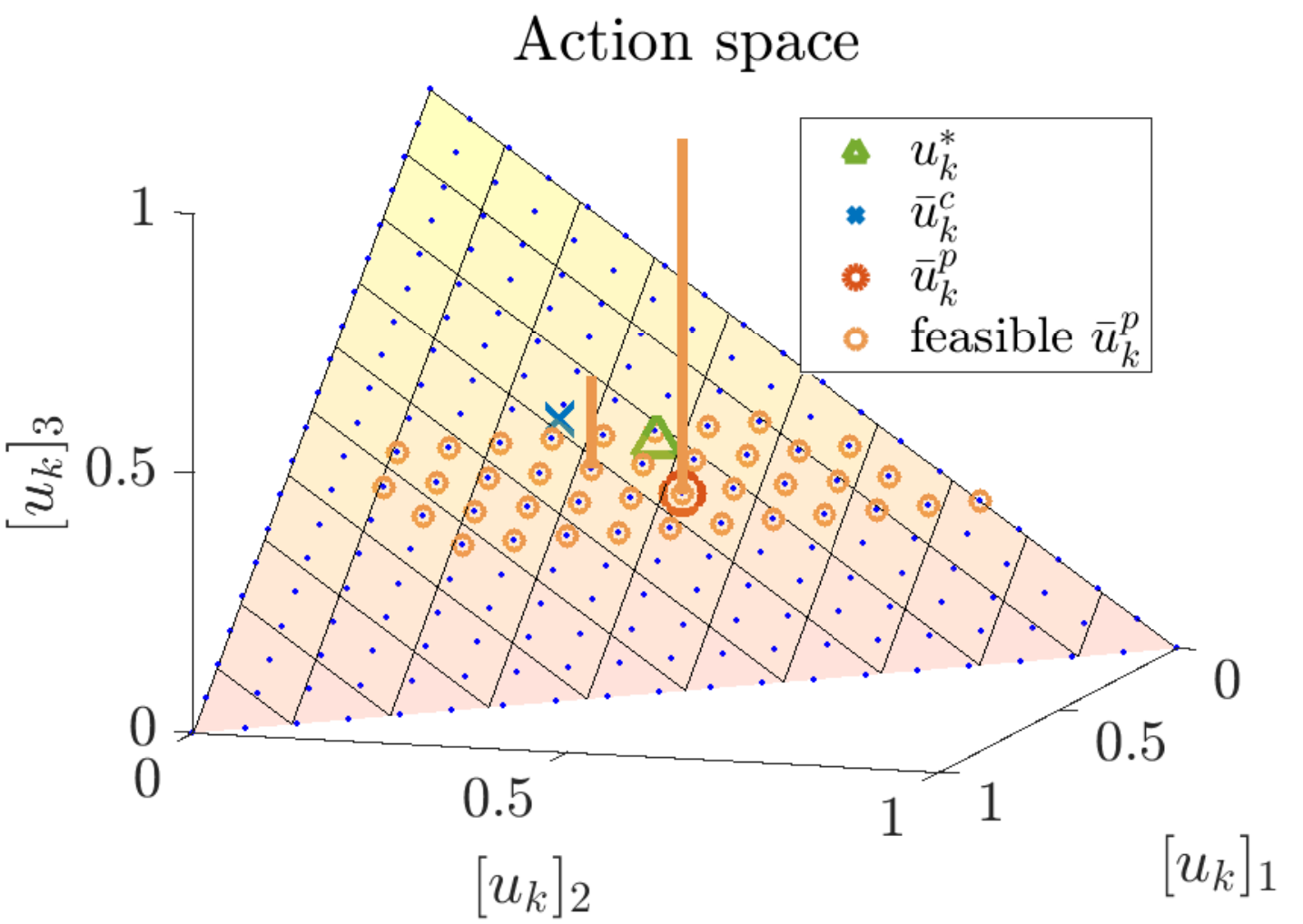}
	\caption{Action space at a certain timestep, showing the actions (\greentriag), (\bluecross), (\orangedot) chosen by the different decision makers. 
For the PDM, the set $\{\bar{u}_k^{(l)}\}_{l=1}^M$ (defined in \eqref{eq:assumption1}) is shown in small circles (\orangesmall), and the probability mass vector $p$ resulting from solving \eqref{eq:assump} is represented by the bars (\orangeline). In this case, two actions had a positive probability of being chosen.} 
	\label{fig:actions}
\end{figure}
\begin{figure}[t!]
	\centering
	\includegraphics[width=0.4\textwidth]{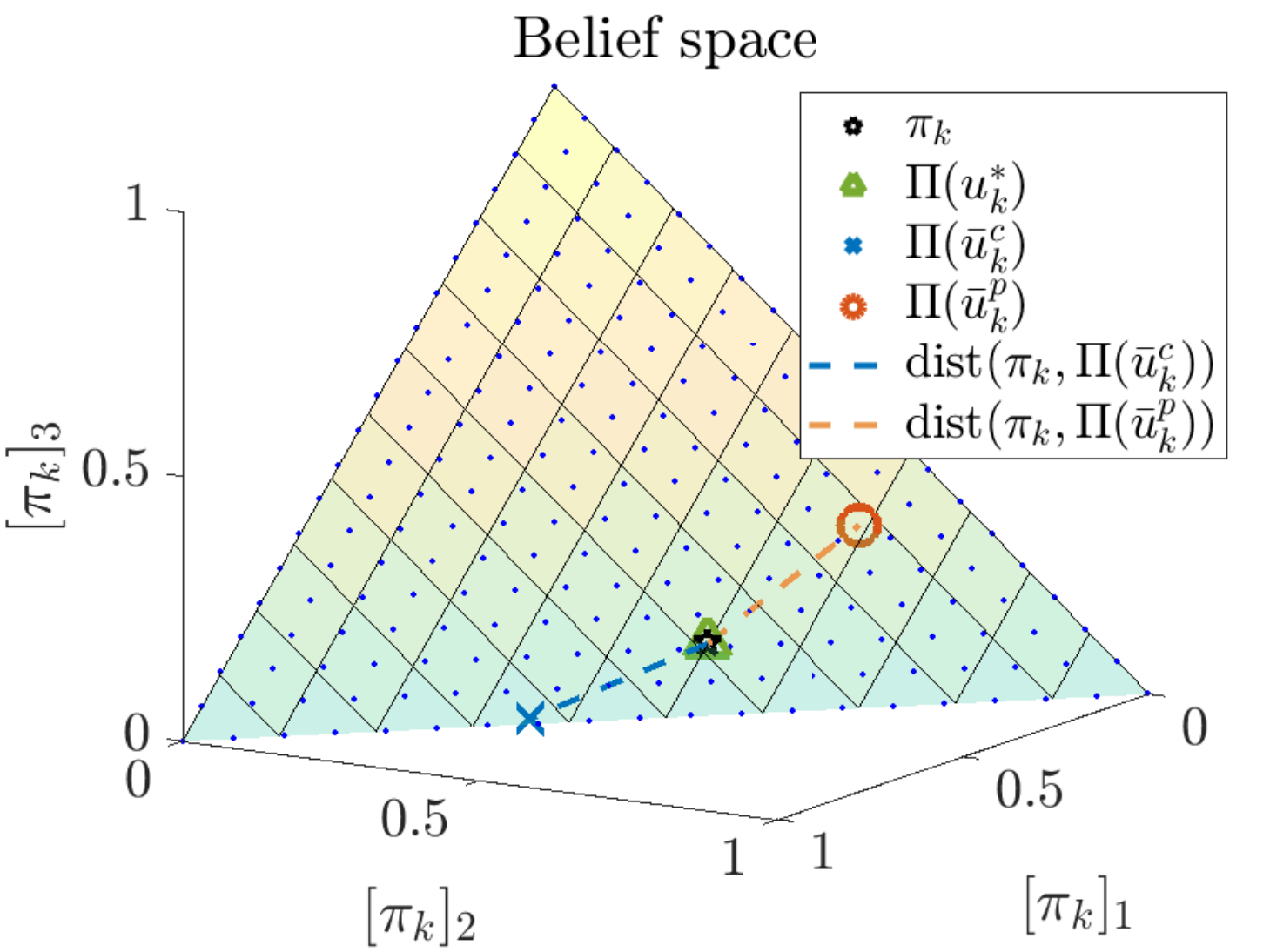}
	\caption{Belief space at a certain timestep. The actual private belief $\pb_k$ is marked in black (\blackstar). The green (\greentriag), blue  (\bluecross) and orange (\orangedot) points are the set of beliefs reconstructed by the adversary, consistent with the ODM, CDM and PDM's chosen action, respectively. In this example, all the sets have a single element, but this is not always the case. The dashed lines measure the privacy level obtained by each agent.}
	\label{fig:beliefs}
\end{figure}

\section{Numerical example}
\label{sec:results}

We now evaluate and visualize our theoretical framework in the adversarial portfolio allocation setup presented in Section \ref{sec:portfolio_allocation}. We first clarify implementation details and then compare the performance of each of the decision makers presented, by analysing their actions, costs, and the set of beliefs that the adversary can reconstruct.

\subsection{Implementation details and conventions}

In order to be able to visualize the results, we randomly generate a three-regime portfolio allocation scenario, where $U=X=3$. This means that both action and belief spaces are represented by two-dimensional unit simplices.

The three other parameters in \eqref{eq:ap1} and \eqref{eq:assump} -- namely, $c_b$, $\Psi(\bar{u}_k)$ and the set $\{ \bar{u}_k^{(l)} \}_{l=1}^M$ -- were chosen as follows. The cost budget was chosen to be $c_b=0.1$, which corresponds to the agent allocating 10\% extra of the optimal cost to obfuscate its private belief. 
The privacy measure $\Psi(\bar{u}_k)$ was chosen to be the measure of \textit{maximal obfuscation}, corresponding to item \ref{item:d} of Section \ref{sec:quantifying}.
Since the set $\Pi(\bar{u}_k)$ is described by convex inequalities, the term $\Psi(\bar{u}_k)$ thus becomes, by definition:
\begin{align}
    \Psi(\bar{u}_k) =  \text{dist}(\pi_k,\Pi(\bar{u}_k)) = \min_{y \in \Pi(\bar{u}_k)}  \left\lVert \pi_k - y \right\rVert_2.
    \label{eq:maxconfuse}
\end{align}       
Finally, the set $\{ \bar{u}_k^{(l)} \}_{l=1}^M$ was generated by computing a regular grid over the simplex -- in the future, more sophisticated sampling could be considered.

Throughout this section, we use the following convention in the figures:
\begin{itemize}

\item The results obtained by the \textit{Original Decision Maker} (ODM), that solves problem \eqref{eq:forward} from Section \ref{sec:preliminaries}, are labeled with (\greentriag); 
\item The results obtained by the \textit{Counter-adversarial Decision Maker} (CDM), that solves problem \eqref{eq:ap1} from Section \ref{sec:confuse}, are labeled with (\bluecross);
\item The results obtained by the \textit{Probabilistic counter-adversarial Decision Maker} (PDM), that solves problem \eqref{eq:assump} from Section \ref{sec:probabilistic}, are labeled with (\orangedot).

\end{itemize}


\subsection{Actions selected by the different decision makers}

Figure \ref{fig:actions} shows the actions chosen by each of the three different decision makers at a particular timestep $k$. The ODM selects action $u_k^*$, which is the optimal action performed if there is no adversary. The CDM, that is employing the generally intractable approach, selects $\bar{u}^{c}_k$. Finally, the PDM takes a random action between those marked as (\orangesmall), where each has a probability given by the vector $p$, here represented as the bar on top of each action. The randomly chosen action is denoted as $\bar{u}^{p}_k$.


\subsection{The adversary's belief estimates}

The sets of beliefs that the adversary can reconstruct from each of the agent's actions at this timestep (as described in Section \ref{sec:adversary}) are illustrated in Figure \ref{fig:beliefs}. The actual private belief $\pi_k$ of the decision makers is shown in black. The privacy of the ODM is compromised, since its private belief belongs to the set of beliefs reconstructed by the adversary ($\pi_k \in \Pi(u_k^*) \Leftrightarrow dist(\pi_k, \Pi(u_k^*)) = 0$).
On the other hand, the suboptimal actions performed by both the CDM and the PDC have allowed them to successfully obfuscate their private belief ($\pi_k \notin \Pi(\bar{u}_k^c)$ and $\pi_k \notin \Pi(\bar{u}_k^p)$).
Nevertheless, according to the \textit{maximal obfuscation} criterion chosen, the level of privacy depends on the distance between $\pi_k$ and $\Pi(\bar{u}_k)$, shown in dashed lines and computed by \eqref{eq:maxconfuse}. 

Figure \ref{fig:distances2} summarizes these distances over multiple timesteps. The PDM is, on average, better at obfuscating its private belief than the CDM. This is due to the latter having the possibility of obtaining a local minimum when solving optimization problem \eqref{eq:ap1}. Moreover, its budget constraint is hard (compared the PDM, which only needs to satisfy the budget \emph{on average)}. Finally, note that the ODM's privacy is always compromised (the distance is zero).

\begin{figure}[t]
	\centering
	\includegraphics[width=0.4\textwidth]{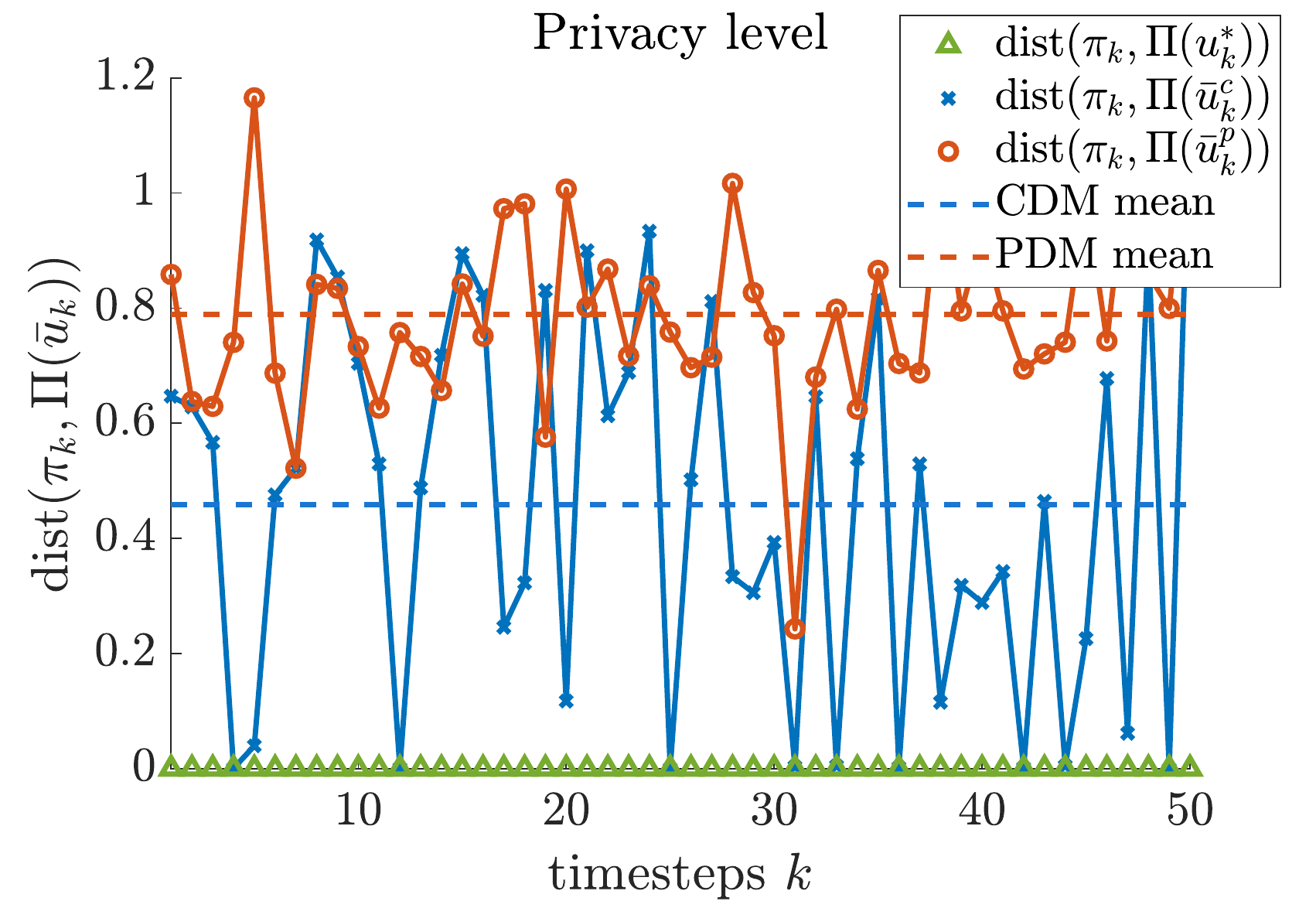}
	\caption{Privacy level of the decision makers at multiple timesteps, measured by the distance between their actual private belief $\pi_k$, and the set of beliefs $\Pi(\bar{u}_k)$ reconstructed by the adversary. On average, the PDM is better at obfuscating its private belief than the CDM, while the ODM's privacy is always compromised.}
	\label{fig:distances2}
\end{figure}


\begin{figure}[t]
	\centering
	\includegraphics[width=0.4\textwidth]{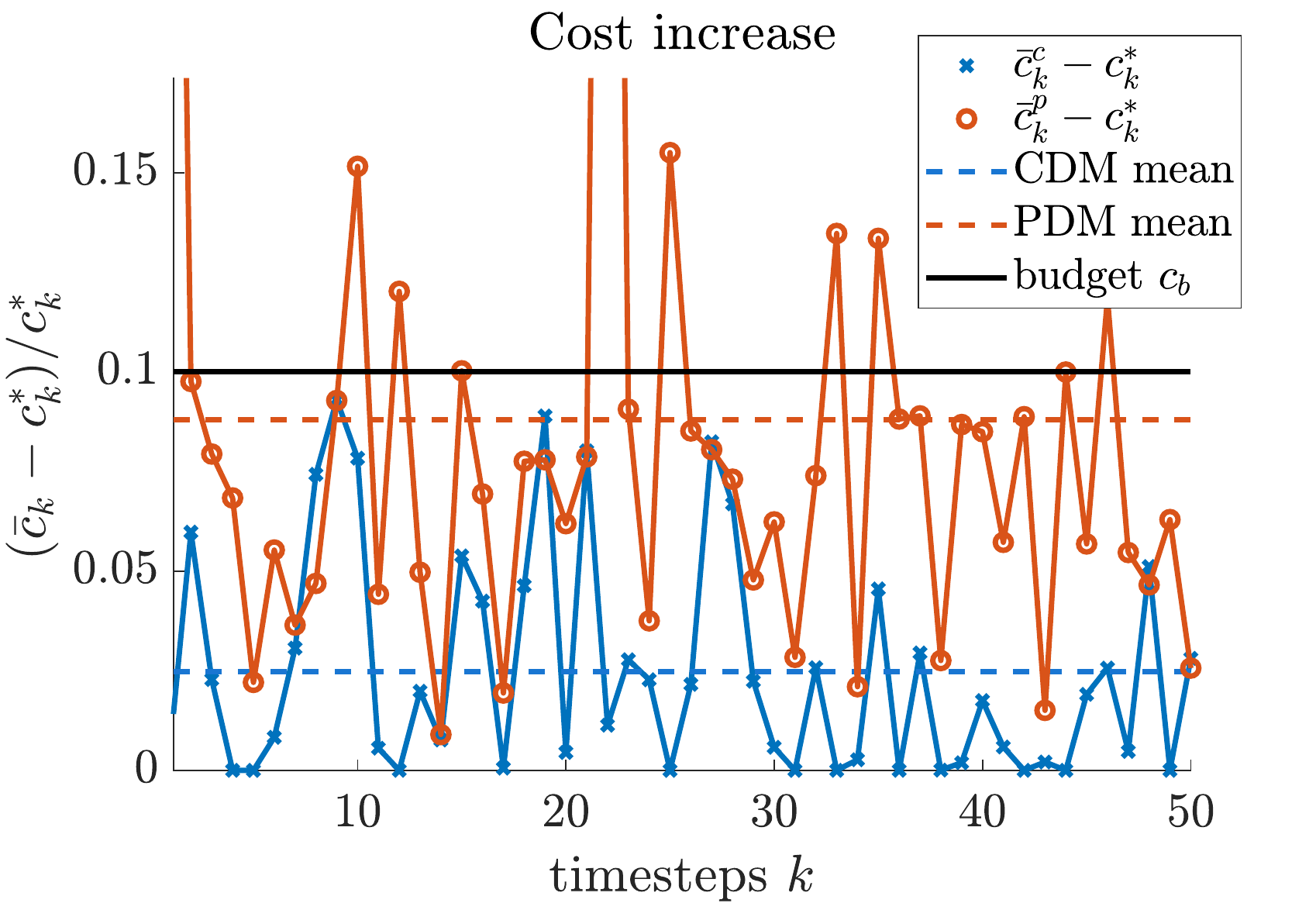}
	\caption{Normalised cost increase for the counter-adversarial decision makers, given by the difference between $\bar{c}_k$ and $c^*_k$, at each timestep. On average, the CDM incurs in a lower cost than the PDM's.}
	\label{fig:costs}
\end{figure}

\subsection{Cost increase for different decision makers }

Above, we saw qualitatively (in Figure \ref{fig:distances2}) that the counter-adversarial decision makers can obfuscate their private belief. We now quantify this and study their increase in cost accrued due to taking the suboptimal action.

Recall that to obfuscate their private belief, the agents allocate a cost budget $c_b$. In Figure \ref{fig:costs}, we show the normalised difference between the costs $\bar{c}_k$ incurred by the two counter-adversarial decision makers, CDM and PDM, and the optimal cost $c_k^*$ incurred by the ODM, over multiple timesteps. It can be seen that the CDM incurs, on average, on a lower cost than the PDM. It should also be noted that while the CDM never violates the cost budget constraint, the PDM fulfills the requirement of not violating it on average.

Figures \ref{fig:distances2} and \ref{fig:costs} show the trade-off between preserving the agent's privacy and limiting its cost. The more the agent aims to protect its privacy, the more it needs to perform an action further from the optimal one. Therefore, as expected, for either of the counter-adversarial approaches, a higher privacy level entails a higher cost, and vice-versa. The fact that the PDM's cost is, on average, higher than the CDM's, is thus explained by the fact that its privacy is, on average, higher. 

\subsection{Influence of the cost budget on the level of privacy }

Recall that the previous results were obtained for a fixed cost budget $c_b$.
As a final result, we show in Figure \ref{fig:tradeoff} how the privacy level of the agents varies with the budget they allocate to obfuscate their private belief. At a certain budget, the privacy level saturates since the furthest action away from the optimal is already selected. We can also see that the PDM saturates at a higher cost and has, on average, a higher privacy than the CDM for any cost budget.

\begin{figure}[t]
	\centering
	\includegraphics[width=0.4\textwidth]{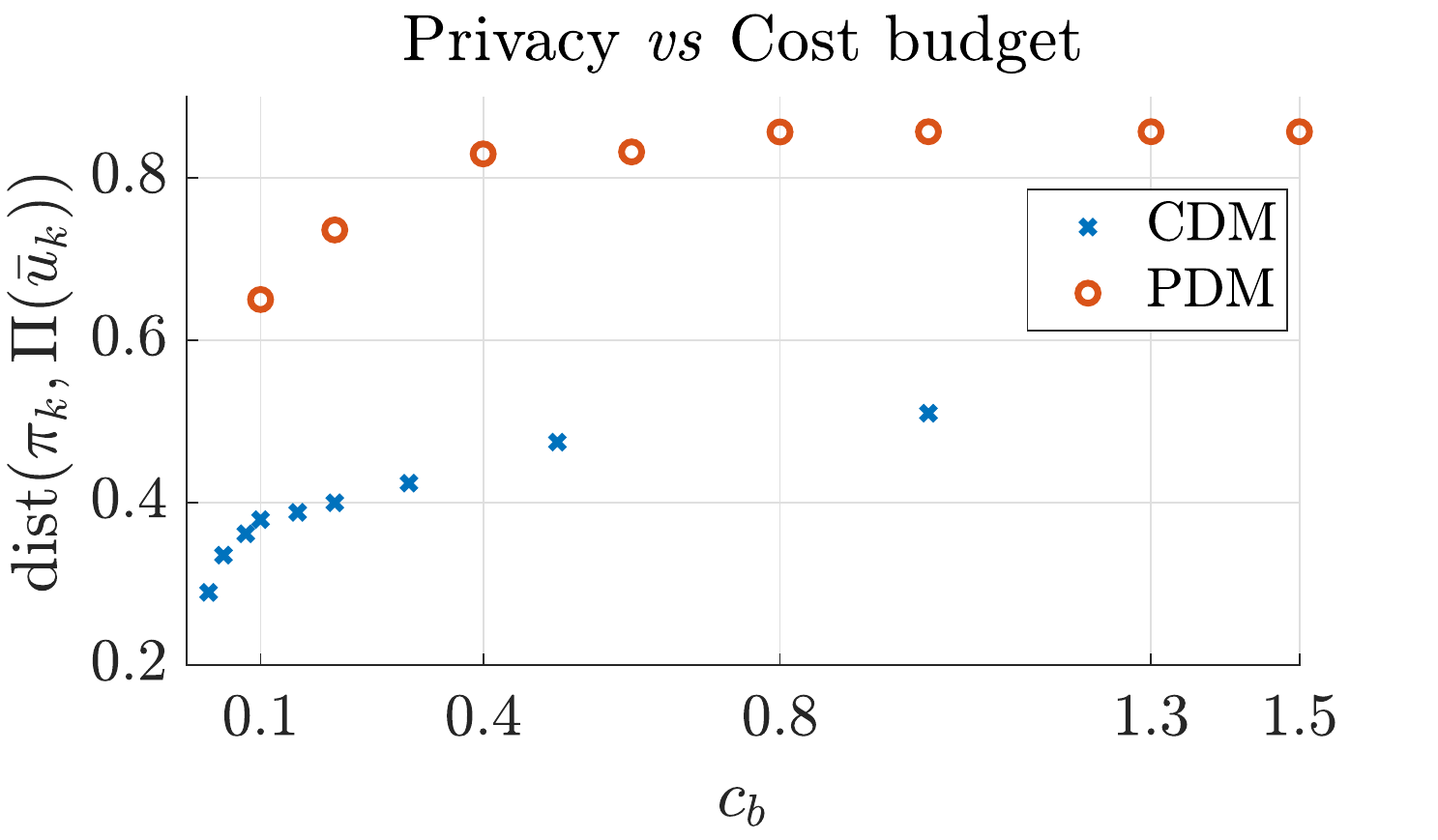}
	\caption{Trade-off between level of privacy and obfuscation cost budget, according to the privacy criterion chosen. The privacy level of the agents increases with their cost budget, until saturation. 
The plot shows an average over 20 simulations for each cost budget, under certain market conditions.}
	\label{fig:tradeoff}
\end{figure}


\section{Conclusions}
\label{sec:conclusions}
In this paper, we studied counter-adversarial decision-making. It has recently been shown that an adversary can reconstruct data private to a decision-making agent by observing its actions. We proposed a framework for trading privacy against cost-optimality: by performing a suboptimal action, the decision maker can conceal its private belief. Several measures to quantify this trade-off were discussed. The solution to the trade-off resulted in an intractable optimization problem, for which we derived a probabilistic relaxation that relies only on linear-programming. The framework and the proposed methods were evaluated in numerical experiments with promising results -- the decision maker was able to obfuscate its private belief from the adversary, while respecting the cost budget allocated to protecting its privacy. 


\subsection{Future Work}
Future work involves relaxing the assumption that the adversary knows the decision maker's cost function to the case where it has to estimate it, as well as investigating problem-specific ways of generating the samples for problem \eqref{eq:assump}. Further, it was shown in \cite{mattila2019smoother} that, if the adversary has knowledge of how the decision maker is updating its private belief, it is possible to infer a full Bayesian posterior. It would be interesting to incorporate counter-adversarial decision-making to this setting.
\balance
\bibliography{confusing-adversary.bib}

\end{document}